\documentclass[sigconf]{acmart}

\usepackage{balance}

\def\BibTeX{{\rm B\kern-.05em{\sc i\kern-.025em b}\kern-.08emT\kern-.1667em\lower.7ex\hbox{E}\kern-.125emX}}
    
\copyrightyear{2023}
\acmYear{2023}
\setcopyright{acmlicensed}\acmConference[SPAA '23]{Proceedings of the 35th ACM Symposium on Parallelism in Algorithms and Architectures}{June 17--19, 2023}{Orlando, FL, USA}
\acmBooktitle{Proceedings of the 35th ACM Symposium on Parallelism in Algorithms and Architectures (SPAA '23), June 17--19, 2023, Orlando, FL, USA}
\acmPrice{15.00}
\acmDOI{10.1145/3558481.3591076}
\acmISBN{978-1-4503-9545-8/23/06}

\usepackage{booktabs,epstopdf} 

\usepackage{multirow,rotating}
\usepackage{varwidth}
\usepackage{caption}
\usepackage[hang,flushmargin]{footmisc} 

\usepackage{amsfonts}
\usepackage{xcolor}
\usepackage{enumitem}
\usepackage{footmisc}
\usepackage{array}
\usepackage{amsmath,amsthm}
\usepackage[ruled,vlined]{algorithm2e}
\usepackage{hyperref}
\newcolumntype{M}[1]{>{\centering\arraybackslash}m{#1}}

\usepackage[normalem]{ulem}

\DeclareMathOperator{\Var}{Var}

\usepackage[usestackEOL]{stackengine}

\newif\ifshowproofs
\showproofstrue

\newtheorem{theorem}{Theorem}
\newtheorem{definition}[theorem]{Definition}
\newtheorem{corollary}[theorem]{Corollary}
\newtheorem{lemma}[theorem]{Lemma}
\newtheorem{construction}[theorem]{Construction}

\begin{document}

\title{
Optimal Round and Sample-Size Complexity for Partitioning in Parallel Sorting}

\author[]{Wentao Yang} 
\authornote{Both authors contributed equally}
\affiliation{
\department{Department of Computer Science}
\institution{University of Illinois Urbana-Champaign}
\streetaddress{}
\city{Urbana} 
\state{IL}
\country{USA}
}
\email{wentaoy2@illinois.edu}

 \author[]{Vipul Harsh}
 \authornotemark[1]
\affiliation{
\department{Department of Computer Science}
\institution{University of Illinois Urbana-Champaign}
\streetaddress{}
\city{Urbana} 
\state{IL}
\country{USA}
}
 \email{vharsh2@illinois.edu}

 \author[]{Edgar Solomonik} 
\affiliation{
\department{Department of Computer Science}
\institution{University of Illinois Urbana-Champaign}
\streetaddress{}
\city{Urbana} 
\state{IL}
\country{USA}
}
\email{solomon2@illinois.edu}







%




\begin{abstract}
State-of-the-art parallel sorting algorithms for distributed-memory architectures are based on computing a balanced partitioning via sampling and histogramming.
By finding samples that partition the sorted keys into evenly-sized chunks, these algorithms minimize the number of communication rounds required.
Histogramming (computing positions of samples) guides sampling, enabling a decrease in the overall number of samples collected.
We derive lower and upper bounds on the number of sampling/histogramming rounds required to compute a balanced partitioning.
We improve on prior results to demonstrate that when using $p$ processors, $O(\log^* p)$ rounds with $O(p/\log^* p)$ samples per round suffice.
We match that with a lower bound that shows that any algorithm with $O(p)$ samples per round requires at least $\Omega(\log^* p)$ rounds. Additionally, we prove the $\Omega(p \log p)$ samples lower bound for one round, thus proving that existing one round algorithms: sample sort, AMS sort~\cite{axtmann2015practical} and HSS~\cite{hssSPAA} have optimal sample size complexity.
To derive the lower bound, we propose a hard randomized input distribution and apply classical results from the distribution theory of runs. 
\end{abstract}

\maketitle



\section{Introduction}

\subsection{Background}
Sorting a distributed array (a common problem and widely used primitive) is challenging due to the communication overhead needed to infer a load-balanced partition of the global order of the data items, especially for sorting large scale data. 
Modern parallel sorting algorithms minimize data movement by first computing an approximately balanced partitioning of the global order, then sending data directly to the destination processor.
Among the simplest such algorithms is \emph{sample sort}~\cite{frazer1970samplesort}, which infers the partitioning by collecting a sample of data items of each processor and by picking evenly spaced samples as splitters that are close to quantiles in the global order.
Sample sort achieves a balanced partitioning with $p$ processors given a sample size of $\Theta(p \log p)$, but such a sample can become expensive to store/process for large $p$, such as in a work by Cheng et al.~\cite{Cheng07anovel}, where $p$ is often in the order of hundred thousands.
Multi-round sampling can be used to lower the total sample size, by \emph{histogramming}~\cite{kale1993comparison,solomonik2010highly} samples and drawing subsequent samples from a subset of the global input range.
Specifically, a \emph{histogram} of a sample gives the rank in the global order of each item in the sample.
For example, the \emph{Histogram sort with sampling} algorithm has been shown to require $O(\log \log p)$ rounds of $O(p)$ samples per round to achieve a balanced partitioning. 

There is a fundamental tradeoff between the total sample size and the total of number of sample-histogram rounds (which can be performed with 2 BSP supersteps). Less rounds are desired, but additional rounds permit more selective sampling and consequently less communication.
The exact relationship of the tradeoff, as well as the optimal round complexity and sample complexity remain open problems.

\subsection{Our contribution}
In this paper, we aim to address these problems. Our contribution is twofold. First, we provide a tighter analysis of Histogram sort with sampling~\cite{hssSPAA}.
We demonstrate that with a total sample size complexity of $O(p)$ using $O(\log^* p)$ rounds\footnote{
We use the following definition of $\log^* x$ for any positive real number $x$: $\log^* x =0$ if $x \leq 1$, $\log^* x=      1 + \log^* (\log x)$  if $x > 1$.
} of sampling followed by histogramming (Section ~\ref{histopartAlg}) suffices to achieve a balanced partition with high probability. 
We refer to the adjusted sampling/histogramming stage of the HSS algorithm as \textit{Histogram Partitioning}.
Second, we show that the sample size and round complexity of the approach are optimal by proving (Section~\ref{lowerBound}) that with $O(p)$ samples per round, $\Omega(\log^* p)$ rounds are necessary for any randomized algorithm to compute a balanced partitioning.
To prove this lower bound, we design a difficult input distribution and apply Yao's principle,  
making use of 
results from distribution theory of success runs to derive the result.
A corollary of our lower bounds is that, since any one round algorithm requires $\Omega(p \log p)$ samples - Sample sort, HSS with one round, AMS-sort \cite{axtmann2015practical} are essentially optimal in terms of sample size complexity. Figure~\ref{fig:tradeoff} illustrates the various tradeoff points between the number of sample-histogram rounds vs the total sample size and our lower bounds. 
 
 
\subsection{Related work}
In data-heavy large scale settings, often it's infeasible to move data multiple times. Hence, large scale parallel sorting algorithms~\cite{axtmann2015practical,hssSPAA,sundar2013hyksort,solomonik2010highly} first determine a set of $p-1$ splitters to split the data into $p$ buckets ($p$ is the total number of processors). The set of splitters are broadcast to all processors and each processor sends keys falling between the $(i-1)^{th}$ and the $i^{th}$ splitter to $i^{th}$ processor. The choice of the splitters directly determines how load balanced the final data is across processors. Often, algorithms~\cite{hssSPAA,axtmann2015practical} provide an extra guarantee that each processor would end up with no more than $\frac{N(1+\epsilon)}{p}$ keys where $N$ is the total number of keys across all processors. Past works on parallel sorting have extensively studied algorithms that efficiently compute the set of splitters (also referred to as a partition) that  guarantee $(1+\epsilon)$ load imbalance. The data-partitioning algorithm is also the focus of this paper. The final part of parallel sorting is the data exchange phase where each processor sends data to each other based on the splitters. Efficient data-exchange algorithms are outside the scope of this paper. We refer the reader to ~\cite{axtmann2015practical} for a study.

Sample sort~\cite{frazer1970samplesort}, an early parallel sorting algorithm achieves $(1+\epsilon)$-balanced partitioning  by sampling keys either randomly~\cite{blelloch1998experimental} or evenly across all processors~\cite{li1993versatility}, requiring a total sample size of $O\big(\frac{p\log p}{\epsilon^2}\big)$. Histogramming is the process of carrying out a reduction across all processors to obtain the ranks of all sampled keys. Histogramming allows the partitioning algorithm to make a more informed choice and reduces the sample size requirements. AMS-sort with 1 sample-histogram round requires $O(p(\log p+\frac{1}{\epsilon}))$ samples. We note that the sample size is directly proportional to the communication cost of the data partitioning step~\cite{hssSPAA}, hence algorithms aim to minimize the sample size. 

One could achieve a significantly lower sample size complexity by repeated rounds of sampling followed by histogramming~\cite{sundar2013hyksort,hssSPAA}. Histogram sort with Sampling (HSS)~\cite{hssSPAA} achieves a sample size complexity of $O(p\log \log p)$ with $O(\log \log p)$ rounds of sampling/histogramming, while Hyksort with a slightly different sampling algorithm achieves a sample size complexity of $O(p\log p)$~\cite{hssSPAA}. Note that in this work, we focus on sample sort and histogram sort-like algorithms, which require a reduction of samples to one processor and a broadcast of data from that processor to others. We refer the reader to ~\cite{hssSPAA} for a detailed review of these algorithms.

The optimal communication cost for parallel sorting algorithms in a general BSP model has been previously studied by Goodrich~\cite{goodrich1999communication}.
Goodrich~\cite{goodrich1999communication} gives an optimal algorithm that requires $O\big(\frac{\log N}{\log{h+1}}\big)$ rounds with each processor sending and receiving at most $h=\Theta\big(\frac{N}{p}\big)$
items per round. Goodrich also proves a tight round lower bound $\Omega\big(\frac{\log N}{\log{h+1}}\big)$.
Our lower bound analysis applies in a more restricted setting, which we discuss in Section~\ref{sec:model}.
This theoretical setting characterizes partition-based parallel sorting algorithms that perform well on current architectures (HSS and HykSort).
Partition-based algorithms have the advantage of moving data (aside from samples) only once or twice, unlike merge-based algorithms such as the one employed by Goodrich.


To the best of our knowledge, our analysis is the first to apply the distribution theory of runs~\cite{mood1940distribution}, which quantify the frequency of sequences of successes in a series of random trials, to analysis of parallel sorting/partitioning.
This connection may be employed for new or simplified analysis of other sampling-based partitioning algorithms for parallel sorting or related problems.
Further information on the topic may be found in prior surveys of topics related to the distribution theory of runs~\cite{balakrishnan2011runs}.

\subsection{Problem Statement}



Let $A(1),\ldots,A(N)$ be an input sequence of keys (without duplicates) distributed across $p$ processing elements, such that each processor has $\frac{N}{p}$ keys. We assume that the keys could be of any data type, so that the algorithm is purely comparison-based. Let $R(k)$ denote the rank of key $k$ in the array $A$ and let $I(r)$ denote the key of rank $r$.
The goal of parallel sorting is to sort the input keys, so that the $i^{th}$ processor owns the $i^{th}$ subsequence of the sorted keys $I(1),\ldots, I(N)$ and that each subsequence consists of at most $(1+\epsilon)\frac{N}{p}$ keys.
The goal of the partitioning algorithm is to determine $p-1$ splitter keys $s_1, s_2, \ldots,  s_{p-1}$ which divide the input range into $p$ buckets. We refer to the set of splitter keys as a partition.
We wish to achieve a well-balanced partition so that the size of each bucket is small. We choose $s_0$ and $s_p$ such that $R(s_0) = 0$ and $R(s_p) = N$ for notational convenience. 

Though we focus on the parallel partitioning problem in this paper, we make references to the parallel sorting problem since parallel partitioning is a crucial part of parallel sorting- if the input data has been partitioned, each processor can simply sort locally which solves the parallel sorting problem. 


\begin{definition} \label{partitionDefinition}
A partition $(s_0 \leq s_1 \leq s_2 ... \leq s_p)$ is $(1+\epsilon)$-balanced if $R(s_{i+1}) - R(s_{i}) \leq (1+\epsilon)\frac{N}{p} \ \ \forall i\in \{0,\ldots, p-1\}.$ 

\end{definition}
Table~\ref{tab:notation} defines further notation used in the algorithm and the lower bound proofs.

\begin{table}[htb]
\centering
\footnotesize
\sf
\begin{tabular}{@{}l|ll@{}}
\toprule
\multirow{5}{*}{\begin{turn}{90}\Shortunderstack{{algorithm}}\end{turn}} 
& $p$ & number of processors sorting keys \\
& $N$ & number of keys to sort \\
& $A(j)$ & the $j^{th}$ input key, $A$ is the set of all keys  \\
& $R(i)$ & initial index of the key ranked $i^{th}$ globally \\
& $I(r)$ & key with rank $r$ in the overall global order\\
& $W_i$ & the set of undetermined splitters are the $i^{th}$ round\\
\hline
\multirow{6}{*}{\begin{turn}{90}\Shortunderstack{{lower bound}}\end{turn}} 
& $C$ &\Shortunderstack{the number of sub-intervals in each interval (same as the number \\of contiguous sequences at the start of the first round)} \\
& $T$ & the number of rounds in the algorithm\\
& $D_i$ & bound on the number of contiguous sequences at round $i$\\
& $Q_i$ & the number of intervals in each contiguous sequence at round $i$\\
& $K_i$ & the number of samples available for each interval on average at round $i$\\

\bottomrule
\end{tabular}
\caption{Notation used in this paper.}
\label{tab:notation}
\vskip -0.65 cm
\end{table}

\subsection{Cost and Execution Model}
\label{sec:model}
In this work, we restrict our attention to sample-sort and histogram-sort like algorithms in the BSP model. This allows us to prove a different lower bound on round complexity than the previously mentioned result by Goodrich in the general setting of parallel sorting in the BSP model~\cite{goodrich1999communication}. Our lower bound also applies to the Map-Reduce model, in which only one-to-all and all-to-one communications are allowed. 


We analyze partitioning algorithms based on the total number of samples collected and the number of sampling/histogramming rounds.
These bounds naturally correspond to round/synchronization complexity and communication cost in multiple models.
One such model is the bulk synchronous parallel (BSP)~\cite{valiant1990bridging}.
In the BSP model, algorithms proceed in \textit{supersteps}\footnote{
    To clarify the notation, a round of sampling and broadcasting consists of two BSP supersteps.
}.
In each superstep,
\begin{itemize}[leftmargin=*]
\item each processor spends a fixed amount of time performing local computation,
\item each processor sends/receives message up to a given total size ($h$) to/from other processors.
\end{itemize}

Our algorithm and lower bound for the cost of partitioning apply in a BSP setting with $h=p$, but also require that only $O(p)$ distinct keys are collected/communicated at each step.
To align our setting with that of Goodrich~\cite{goodrich1999communication}, we would need $h=\Theta(p)=\Theta(n/p)$, in which case Goodrich's algorithm requires a constant number of rounds, as it leverages communication/exchange $O(hp)$ total keys per round instead of $O(h)$.
We further note that our lower bound applies to the weaker problem of finding a $(1+\epsilon)$-balanced partitioning as opposed to complete parallel sorting.


Our results also lead to complexity bounds for algorithms in the massively parallel computing (MPC) model~\cite{mapReduceModel}. 
The MPC model is very similar to the BSP model, but it focuses on bounding the number of rounds in a distributed algorithm, given a limit on the amount of memory.
In partitioning algorithms $\Theta\big(\frac{N}{p}+p\big)$ memory is typically used to store local data and the partition.
Our algorithm requires $O\big(\frac{N}{p}+p\big)$ memory and $O(\log^* p)$ rounds, and our tight lower bound still holds in the setting given similar restrictions.
This constitutes a reduction in memory footprint by $\Theta(\log p)$ with respect to sample sort.
Figure~\ref{fig:tradeoff} provides a summary of the lower and upper bounds we introduce, as well as past work (Sample sort, HSS~\cite{hssSPAA} and AMS-sort~\cite{axtmann2015practical}).

\noindent \textbf{Collective operations}: In our algorithm, we make use of collective operations that are extensively used in parallel algorithms. We refer to figure 1 of ~\cite{collectiveOps} for a quick overview of some common collective operations. Specifically, in the BSP model, broadcasts of $O(p)$ elements and reductions of $O(p)$ elements from each process, are efficiently done via scatter or reduce-scatter, respectively, followed by all-gather, both of which have communication cost $O(p)$ and require 1 BSP superstep.



\section{SAMPLE ROUND COMPLEXITY UPPER BOUND AND FAST ALGORITHM} \label{histopartAlg}

\begin{figure}
\begin{center}
\includegraphics[width=0.3\textwidth]{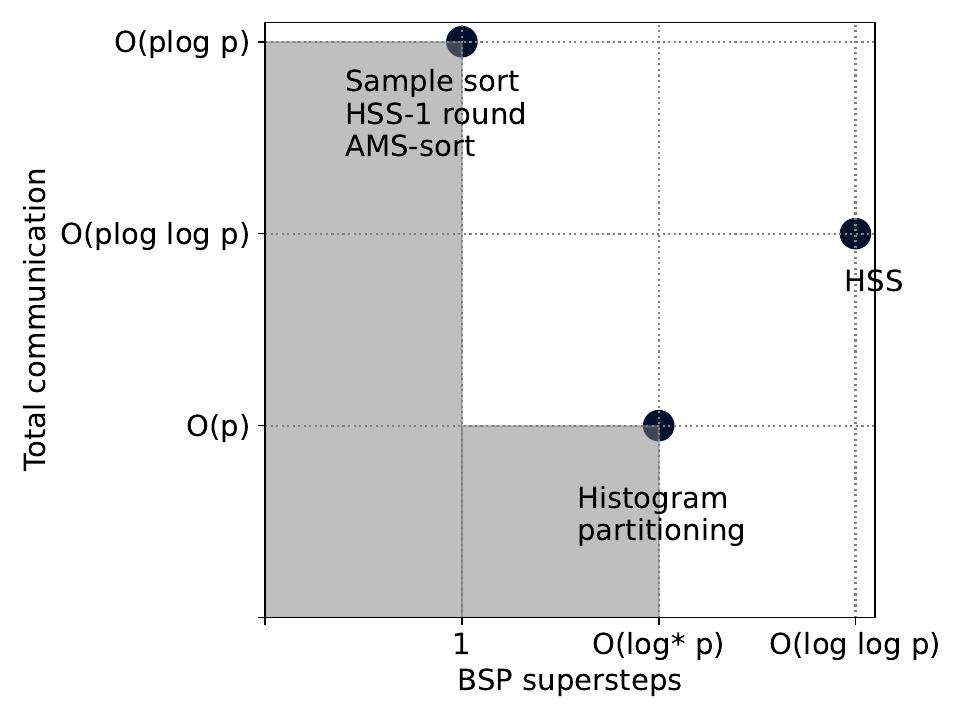}
\caption{Complexity of histogram/sample-sort approaches in BSP supersteps and total communication. Histogram partitioning is the sampling/histogramming approach introduced and analyzed in this paper. The shaded region represents our lower bounds, no algorithm based on global sampling can be strictly inside the region.}
\label{fig:tradeoff}
\end{center}
\end{figure}


In this section, we review the HSS algorithm~\cite{hssSPAA}, expressed for a more general sample size.
We then analyze the histogramming/sampling stage with a different choice of sampling ratios, which yields our Histogram Partitioning approach.
For any constant $\epsilon > 0$, when seeking a $(1+\epsilon)$-balanced partitioning, the new analysis shows that $O(\log^* p)$ rounds and $O(p)$ communication volume suffice (opposed to $O(\log \log p)$ rounds and $O(p \log \log p)$ communication volume, as analyzed in ~\cite{hssSPAA}).

\subsection{Histogram Sort with Sampling} \label{HP alg}
We first describe the Histogram Sort with Sampling (HSS) algorithm from~\cite{hssSPAA} for a general sampling ratio.
In each round, the HSS algorithm computes a histogram of a sample of keys.
A histogram of keys $a_1,\ldots, a_q$ is given by their global ranks $R(a_1),\ldots, R(a_q)$.
The global ranks are computed by summing (reduction) of the local ranks of the keys $a_1,\ldots, a_q$ within the set of keys assigned to each processor.
Each histogram serves to narrow down the range of keys out of which the subsequent samples are collected, with the aim of finally finding keys $a_1,\ldots, a_{p-1}$, whose global ranks give an accurate (load-balanced) splitting, namely $R(a_i)\in [i\frac{N}{p}, i(1+\epsilon)\frac{N}{p}]$.
We refer to such an $a_i$ as a 'determined' splitter.

In general, given two keys $a$, $b$ with $a<b$, we define an interval $[a,b]$ as the set of keys in the input sequence that are greater than $a$ and smaller than $b$.
HSS maintains bounds for each of the $p-1$ splitters, which are successively improved.
In the $i^{th}$ round the bounds for splitter $j$ are $[L_i(j),U_i(j)]$, with $L_1(j)=I(1)$, $U_1(j)=I(N)$ for all $j\in\{1,\ldots, p-1\}$.
The set of undetermined splitters is at round $i$ is $W_i$.
The set of keys from which HSS collects a sample is given by the union of the bounds for all splitters which the algorithm has not yet determined.
For a choice of sampling ratio $\psi_i$, the HSS algorithm computes the following steps at the $i^{th}$ round (the pseudocode is also given in the Appendix as Algorithm~\ref{hss:alg}):
\begin{itemize}
\item if all $p-1$ splitters have been determined, broadcast them, redistribute keys so that processor $i$ receives keys between the $i^{th}$ and $(i+1)^{th}$ splitters, sort locally, and terminate, otherwise,
\item each processor collects a local sample by 
sampling each local key in $\gamma_i = \bigcup_{j\in W_i}[L_j,U_j]$ with probability $\psi_i$,
\item the combined sample $a_1,\ldots, a_\eta$ is gathered globally, (e.g., by gathering all sample keys to one processor then broadcasting to all processors),
\item each processor computes a local histogram of the sample,
\item a global histogram $H=\{h_1,\ldots, h_\eta\}$, $h_j=R(a_j)$ is computed by reduction of local histograms,
\item if any $H\ni h_j\in [\eta\frac{N}{p}, \eta(1+\epsilon)\frac{N}{p}]$, the corresponding sample $a_j$ is a determined $\eta^{th}$ splitter,
\item for all undetermined splitters, $\eta\in W_{i+1}$, update the splitter bounds using the new histogram, 
\begin{align*}
L_{i+1}(\eta)&=\max\Big(L_{i}(\eta),\max_{h_j\in H, h_j<\eta \frac{N}{p} }(h_j)\Big)\quad \text{and} \\ 
U_{i+1}(\eta)&=\min\Big(U_i(\eta),\min_{h_j\in H, h_j>\eta \frac{N}{p}(1+\epsilon)}(h_j)\Big),
\end{align*}
\item broadcast updates to the splitter bounds, namely the change to $W_{i+1}$, $W_{i}\setminus W_{i+1}$ and $[L_j,U_j]$ for all $j\in W_{i+1}$.
\end{itemize}
In~\cite{hssSPAA}, the sampling ratio at each round is determined by seeking a total sample of a fixed size ($O(p)$) at each step.
In this work, we consider a constant $\epsilon \leq 1$ and set $\psi_i = \frac{p^2}{|W_i|N\log^* p}$.
Since $\gamma_i$ is composed of $|W_i|$ undetermined intervals as well as (due to $\epsilon\leq 1$) at most 1 determined interval before and after each undetermined interval, $|\gamma_i \cap A| \leq 3|W_i|N/p$.
Hence, the expected number of samples in the $i^{th}$ round is  $\psi_i|\gamma_i \cap A|\leq 3\psi_iN|W_i|/p = \frac{3p}{\log^* p}$.
This choice allows us to obtain a total communication volume of $O(p)$, since $O(p)$ communication suffices for both sampling and collecting updated splitter information on all processors.
In the next subsection, we show that w.h.p.\ $O(\log^* p)$ rounds of HSS suffice with this choice of sample ratio.
If aiming to achieve a small (non-constant) $\epsilon$, another $O(\log(\frac{1}{\epsilon}))$ 
histogramming rounds with $O(p)$ communication volume per round are required, following the approach/analysis in~\cite{hssSPAA}.



\subsection{New Running Time Upper Bound}

We show that the modified HSS algorithm (which we refer to as Histogram Partitioning), finishes in $O(\log^* p)$ rounds and requires $O(p)$ samples in total (stated in Theorem~\ref{new-algorithm bound}). 
Our analysis is based on the intuition that the algorithm can be interpreted as a randomized string cutting process. We can think of the global input of keys (in sorted order) as a string of length $N$. Sampling a random key of rank $r$ is equivalent to cutting the string at the $r^{th}$ of $N$ points.
Each substring that is long enough yields an independent partitioning subproblem.
Every Histogram Partitioning round cuts each substring in as many random points as the number of samples selected from that subrange.


Our analysis is based on 2 phases.
In the first phase, there are still few substrings, so the sampling density is low.
In the second phase, many splitters have already been determined, so the total length of all substrings is relatively small.
We now define these phases formally and show that both complete in $O(\log^* p)$ rounds.

We split the iterations of Algorithm~\ref{HP alg} into two phases. In the first phase, we show the partitioning algorithm makes progress so that $O\Big(\frac{p}{\log^ *p}\Big)$ new splitters are determined in each iteration (Lemma~\ref{phase 1 progress}). In the second phase, in each iteration the number of undetermined  splitters goes down quickly according to the $\log^*$ function (Lemma~\ref{undeterminedSplittersExponentialDecrease}). 
\begin{itemize}
\item Phase 1: Iterations where the number of undetermined splitters exceeds $\frac{p}{\log^* p}$ before the start of the iteration.
\item Phase 2: Iterations where the number of undetermined splitters is at most $\frac{p}{\log^* p}$ before the start of the iteration.
\end{itemize}
Since the total number of samples in each round is fixed ($\frac{3p}{\log^* p}$ in expectation), we only need to bound the number of iterations for a complete analysis. We show that both phases of the algorithm finish in $O(\log^* p)$ iterations with high probability (w.h.p.). 



Our analysis keeps track of the of the size of the domain, $\gamma_j$, of input keys that are sampled from at the $j^{th}$ step.
The size of $\gamma_j$ decreases whenever adjacent splitters are determined.
Hence, if most splitters have been determined, the length of the input cannot be much larger than the set of intervals corresponding to undetermined splitters.

We first derive a lower bound for the probability of a previously undetermined splitter being determined in round $j$. Note that the sampling method ensures that sampling in disjoint intervals is independent.
\begin{lemma} \label{phase 1 progress}
If at the start of the $j^{th}$ step of the Histogram Partitioning algorithm, $k=|W_j|\geq \frac{p}{\log^*p}$ splitters are undetermined (Phase 1), then the probability that a given undetermined splitter is determined during the $j^{th}$ step is at least $\frac{p}{2k\log^*p}$.
\end{lemma}
\begin{proof}
The $j^{th}$ round of Histogram Partitioning samples elements in $\gamma_j$ with probability $\psi_j= p^2 /(|W_j|N\log^*p)$.
The probability that any previously undetermined splitter is determined is then given by the probability that at least one of the $\frac{N}{p}$ keys in the corresponding interval is sampled,
\begin{align*}
1 - \Big(1 - \psi_j\Big)^\frac{N}{p} =& 1 - \Big(1 - \frac{p^2}{|W_j|N\log^* p}\Big)^\frac{N}{p} 
=  1 - \Big(1 - \frac{p^2}{kN\log^* p}\Big)^\frac{N}{p} \\
\geq & 1 - e^{\frac{-p}{k\log^* p}} 
\geq  \frac{p}{k\log^* p} - \frac{1}{2}\big(\frac{p}{k\log^* p}\big)^2 \\
\geq & \frac{p}{k\log^* p} - \frac{1}{2}\big(\frac{p}{k\log^* p}\big) = \frac{p}{2k\log^* p}.
\end{align*}
The last inequality uses that $\frac{p}{k\log^* p} \leq 1$, which follows from $k \geq \frac{p}{\log^* p}$, since the algorithm is in Phase 1.
\end{proof}

Lemma~\ref{phase 1 progress} allows us to bound the number of iterations in phase 1.
\begin{lemma} \label{phase 1 bound}
Phase 1 of the algorithm finishes (fewer than $\frac{p}{\log^* p}$ samples remain undetermined) after $O(\log^* p)$ iterations w.h.p..
\end{lemma}
\begin{proof}
We use a multiplicative Chernoff's bound to show that in each iteration, at least $\frac{p}{3\log^* p}$ new splitters are determined w.h.p..
Specifically, let $X$ be the number of determined splitters in the $j^{th}$ round of phase 1. 
We have $X = X_1  + X_2 \ldots + X_k$ where $X_i$ is an indicator random variable with $X_i = 1$ if the $i^{th}$ undetermined splitter is determined in the $j^{th}$ round. Note that the sampling method ensures that sampling in disjoint intervals are independent. Hence, all $X_i$'s are independent of each other. 
From Lemma~\ref{phase 1 progress}, we have $\mathbb{E}[X] \geq k \frac{p}{2k\log^* p} = \frac{p}{2\log^* p}$, where $k$ is the number of undetermined splitters before the $j^{th}$ round. 
Using multiplicative Chernoff bounds, we get
\begin{align*}
    \Pr[X < \frac{p}{3\log^* p}]  =&   \Pr\Big[X < \frac{p}{2\log^* p}\big(1 - \frac{1}{3}\big)\Big] \leq e^{-\frac{p}{2\log^* p}\frac{1}{3^2}},
\end{align*}
which gives the high probability bound.

\end{proof}



We now bound the number of iterations needed for phase 2 to find all remaining undetermined splitters.
In phase 2, the number of undetermined splitters are less than or equal to $ \frac{p}{\log^* p}$.
We show that the number of undetermined splitters remaining after a round decreases exponentially with the gap from $\frac{p}{\log^* p}$.
\begin{lemma} \label{undeterminedSplittersExponentialDecrease}
If the number of undetermined splitters is $k \leq \frac{p}{t\log^* p}$ with $1 \leq t \leq \log \log p$ at the start of the $j^{th}$ iteration of the Histogram Partitioning algorithm, then the number of undetermined splitters after the $j^{th}$ iteration is at most $k/2^t$, w.h.p..
\end{lemma}
\begin{proof}
Let $X_i$ be an indicator random variable, with $X_i = 1$ denoting that the $ith$ splitter remains undetermined after the iteration. 
All $X_i$'s are independent. 
To show the result in the theorem, we seek to bound $X = X_1 +\cdots +X_k$.
The probability that the $i^{th}$ previously undetermined splitter remains undetermined is 
\begin{align*}
\Pr[X_i = 1] = \Big(1 - \psi_j\Big)^\frac{N}{p} =& \Big(1 - \frac{p^2}{kN\log^* p}\Big)^\frac{N}{p} \leq  \Big(1 - \frac{pt}{N}\Big)^\frac{N}{p} \\
\leq & e^{-t}.
\end{align*}
Hence, $\mathbb{E}[X] \leq ke^{-t}$.
Since our sampling strategy is such that sampling within splitter intervals is independent, we can apply a multiplicative Chernoff bound to bound $X$ with high probability. 

Using Chernoff's bound\footnote{We use the multiplicative Chernoff bound, $P[X > (1+\epsilon)\mathbb{E}[X]] \leq e^{\frac{-\epsilon^2\mathbb{E}[X]}{2}}$.}, we get
\begin{align*}
\Pr[X \geq  k2^{-t}] =& \  \Pr\bigg[X \geq \mathbb{E}[X] \bigg(1 + \frac{ k2^{-t}-\mathbb{E}[X]}{\mathbb{E}[X]}\bigg)\bigg] 
\leq \ e^{-\frac{B^2\mathbb{E}[X]}{2}},
\end{align*}
where $B = \frac{ k2^{-t}-\mathbb{E}[X]}{\mathbb{E}[X]}$. 
Since, $ke^{-t}/\mathbb{E}[X] \geq 1$,
\begin{align*}
B^2\geq &((e/2)^{t} - 1)^2\geq 1/9.
\end{align*}
Thus, the probability that $X$ is less than the bound is small, since $e^{-\frac{B^2\mathbb{E}[X]}{2}} = O(e^{-\mathbb{E}[X]}) = O(e^{-p/(\log p \log^*p)})$.
\end{proof}

Lemma~\ref{undeterminedSplittersExponentialDecrease} shows that the number of undetermined splitters goes down rapidly from $\frac{p}{t\log^* p}$ to $\frac{p}{t2^t\log^* p}$, which immediately leads us to an upper bound for the number of iterations in phase 2 to achieve all splitters.

\begin{lemma} \label{phase 2 bound}
Phase 2 of the Histogram Partitioning algorithm (iterations after fewer than $\frac{p}{\log^* p}$ splitters remain undetermined) finishes (all splitters are determined) after $O(\log^* p)$ iterations w.h.p..
\end{lemma}
\begin{proof}

The number of undetermined splitters at the beginning of phase 2 is at most $\frac{p}{\log^* p}$. By Lemma~\ref{undeterminedSplittersExponentialDecrease}, after $k$ iterations, at most\footnote{$\uparrow \uparrow$ denotes tetration, the inverse operation of $\log^*$.} $\frac{p}{(2 \uparrow \uparrow k) \log^* p}$  undetermined splitters remain w.h.p.. Therefore, after $O(\log^* p)$ iterations in phase 2, all splitters are determined w.h.p.. If the number of undetermined splitters becomes smaller than $p/(\log p \log^* p)$ at any point, then all splitters are determined after one more round using the analysis of HSS for one round~\cite{hssSPAA}.
\end{proof}


We can now combine the iteration bounds for the two phases to conclude with our main theorem.

\begin{theorem}\label{new-algorithm bound}
The Histogram Partitioning algorithm with $\epsilon=1$ finds a $2$-balanced partition in $O(\log^* p)$ BSP steps (and an additional $O(\log(\frac{1}{\epsilon}))$ steps for any $\epsilon > 0$) and $O(p)$ total communication w.h.p..
\end{theorem}
\begin{proof}

Combining Lemma~\ref{phase 1 bound} and Lemma~\ref{phase 2 bound}, we can assert that Histogram Partitioning determines all splitters after $\log^* p$ iterations, w.h.p..
The communication in each round of the algorithm consists of gathering a sample of size $O(\frac{p}{\log^* p})$ and computing a histogram via a reduction of size $O(\frac{p}{\log^* p})$.
The bounds defining $\gamma_j$ may be updated locally (redundantly on all processors) or with  $O(\frac{p}{\log^* p})$ communication of distinct updated bounds.
Consequently, the communication cost of the algorithm in the BSP model is $O(p)$. For any $\epsilon > 0$, refer to the analysis in ~\cite{hssSPAA}.

\end{proof}


\section{Partitioning Round Complexity Lower Bound} \label{lowerBound}

In this section, we prove that any randomized algorithm for finding a $2$-balanced partition, that communicates $O(p)$ total keys per round, requires $\Omega(\log^*p)$ rounds to achieve a constant probability of success.
Our result implies that any algorithm needs at least $\Omega(\log^*p)$ rounds to find a $\Delta$-balanced partition with $O(p)$ samples per round, thus giving us the $\Omega(\log^*p)$ sample round complexity lower bound for partitioning with $O(p)$ samples per round.

The proof is structured as follows (see Figure~\ref{fig:lb_outline} in the Appendix for an overview),
\begin{itemize}[leftmargin=*]
\item we use Yao's principle to argue that the probability that for the worst-case input, any randomized algorithm cannot find a $2$-balanced partition is greater than the probability that the best deterministic algorithm cannot find a $2$-balanced partition for any randomized input distribution;
\item we propose a randomized input distribution that we show is hard for any deterministic partitioning algorithm; 
\item to prove the hardness of our proposed input distribution, we invoke a theorem from the distribution theory of runs to upper bound the amount of progress that can be made in each round by any deterministic algorithm;
\item we use an inductive argument to give a concrete lower bound for the probability that a certain number of splitters remain undetermined in each round given $T = \frac{\log^* p}{6}$ total rounds;
\item finally, we use an union bound on the probabilities for each round to argue that with high probability, $2$-balanced partition cannot be determined for large enough $p$. 
\end{itemize}

Our lower bound proof can also characterize all deterministic algorithms that use the information of $O(p)$ samples (for any way of choosing the samples) at each round.
We restrict our attention to asymptotic analysis and assume the number of keys ($N$) is large.

\subsection{Complexity Bounds from Sampling Bounds}

We now show that it suffices to consider deterministic histogram-based algorithms on a particular distribution of inputs.
\subsubsection{Reduction to Histogram-based Algorithms}

We start by showing that, having restricted to communication of $O(r)$ keys per round, comparison-based algorithms cannot perform significantly better than histogram-based algorithms (for some choice of keys to histogram at each round).
First, we show that obtaining a 2-balanced partition is essentially as hard as the global-partitioning obtained by Histogram sort defined below.
\begin{definition} \label{globalPartitionDefinition}
A partition $(s_0 \leq s_1 \leq s_2 ... \leq s_p)$ is $(1+\epsilon)$-globally-balanced if $R(s_{i}) \in [i\frac{N}{p},  (i+\epsilon)\frac{N}{p}]\ \ \forall i\in \{1,\ldots, p-1\}.$ 
\end{definition} Note the difference from a $(1+\epsilon)$-balanced partition defined in Definition~\ref{partitionDefinition}. A $(1+\epsilon)$-globally-balanced partition is also a $(1+\epsilon)$-balanced partition.

\begin{lemma}
\label{lem:onetoall_to_sample}
Given a distribution of inputs, let $T(p)$ be the expected round complexity of the best deterministic comparison-based 2-balanced partitioning BSP algorithm that communicates at most $r \geq p$ keys in each round.
Then, there exists a 2-globally-balanced partitioning algorithm with expected round complexity $T(2p)$ that communicates $2r$ keys per round. 
\end{lemma}
\begin{proof}
We may obtain a 2-globally-balanced partitioning by performing 2-balanced partitioning with $2p$ virtual processors, with each of $p$ processors simulating 2 virtual processors.
In a 2-balanced partitioning with $2p$ processors, no processor may be assigned more than $\frac{N}{p}$ elements, hence the $2p-1$ splitters defining the partitioning must include a splitter from each of the $p$ intervals, which together yield a 2-globally-balanced partitioning.
Execution with $2p$ processors can be done in $T(2p)$ rounds and at most $2r$ communication per round (ensuring the theorem assumption, $2r \geq 2p$), yielding a round complexity of $T(2p)$.
\end{proof}
Next, we show that, in our setting, any algorithm that communicates a given set of keys might as well compute their histogram (global positions), and hence it suffices to consider histogram-based algorithms to derive a round complexity lower bound.
We note that a sample-histogram round requires 2 BSP supersteps (one to collect the sample and one to reduce the histogram).
\begin{lemma}
\label{lem:onetoall_to_sample2}
Given a distribution of inputs, the round complexity of any deterministic comparison-based 2-globally-balanced partitioning BSP algorithm that communicates $O(p)$ keys in each step is no better than the number of sample-histogram rounds needed by a Histogram sort that computes the position of the same $O(p)$ keys in each round.
\end{lemma}
\begin{proof}
Consider the set of comparisons between keys computed at each round by the given algorithm.
Aside from comparisons of pairs of keys assigned to the same processor initially, these comparisons must involve only previously communicated keys.
Histogram sort computes the position in the global order of all communicated keys, which amounts to performing comparisons between all keys and the communicated keys.
Consequently, the given algorithm can infer that a 2-globally-balanced partitioning has been found only if at a given round, the previously communicated keys include a key from each interval, in which case Histogram sort also completes.
\end{proof}

\subsubsection{Reduction to Sampling Algorithms}

Yao's principle allows us to obtain a lower bound on the number of rounds needed by any randomized partitioning algorithm for the worst-case input, by reduction to the expected round complexity of the best deterministic algorithm for a randomized distribution of inputs.
We first state Yao's principle in general form then describe our case as a corollary.
\begin{lemma}[Yao's principle]
\label{YaoPrinciple}
Let $F_T$ denote the space of possible inputs for a problem, $\mathcal{R}_T$ denote the distribution of the outcomes of any randomized algorithm, $A_T$ denote the space of possible deterministic algorithms for the problem, and $\mathcal{D}_T$ denote the input distribution. By Yao's principle, for any cost function $\Theta$ that is a real-valued measure of a deterministic algorithm on an input, we have\footnote{
   Here $\mathbb{E}_{R \leftarrow \mathcal{R}_T}$ denotes the expectation with random variable $R$ taken over distribution $ \mathcal{R}_T$.
}
 \begin{equation}
     \max_{f \in F_T} \mathbb{E}_{R \leftarrow \mathcal{R}_T}[\Theta(R, f)] \geq \min_{A \in A_T} \mathbb{E}_{f \leftarrow \mathcal{D}_T}[\Theta(A, f)]
.\end{equation}
\end{lemma}
\begin{corollary}
\label{YaoPrincipleCoro}
Consider the problem of finding $2$-partition via a comparison-based BSP algorithm that communicates $O(p)$ total keys per round. Consider a given input distribution $\Tilde{\mathcal{D}}_T$.
Also, let $X(R, f)$ define the event that $R$ cannot find $2$-partition for $f$ in $T$ rounds, and let $X(A, f)$ define the event that $A$ cannot find $2$-partition for $f$ in $T$ rounds. Then, we have 
 \begin{equation}
     \max_{f \in F_T} \Pr_{R \leftarrow \mathcal{R}_T}[X(R, f)] \geq \min_{A \in A_T} \Pr_{f \leftarrow \Tilde{\mathcal{D}}_T}[X(A, f)]
.\end{equation}
\end{corollary}
\begin{proof}
After setting cost of achieving $X(R, f)$ and $X(A, f)$ as $0$, and $1$ otherwise, the corollary trivially follows from Lemma~\ref{YaoPrinciple}.
\end{proof}
As a result, we only need to show that $\min_{A \in A_T} \Pr_{f \leftarrow \Tilde{\mathcal{D}}_T}[X(A, f)]$, which is a probabilistic lower bound for any deterministic algorithm on our input distribution, is large enough.
Such a bound gives us a probabilistic lower bound for the original problem.

\begin{figure}
\begin{center}
\includegraphics[width=1.\linewidth]{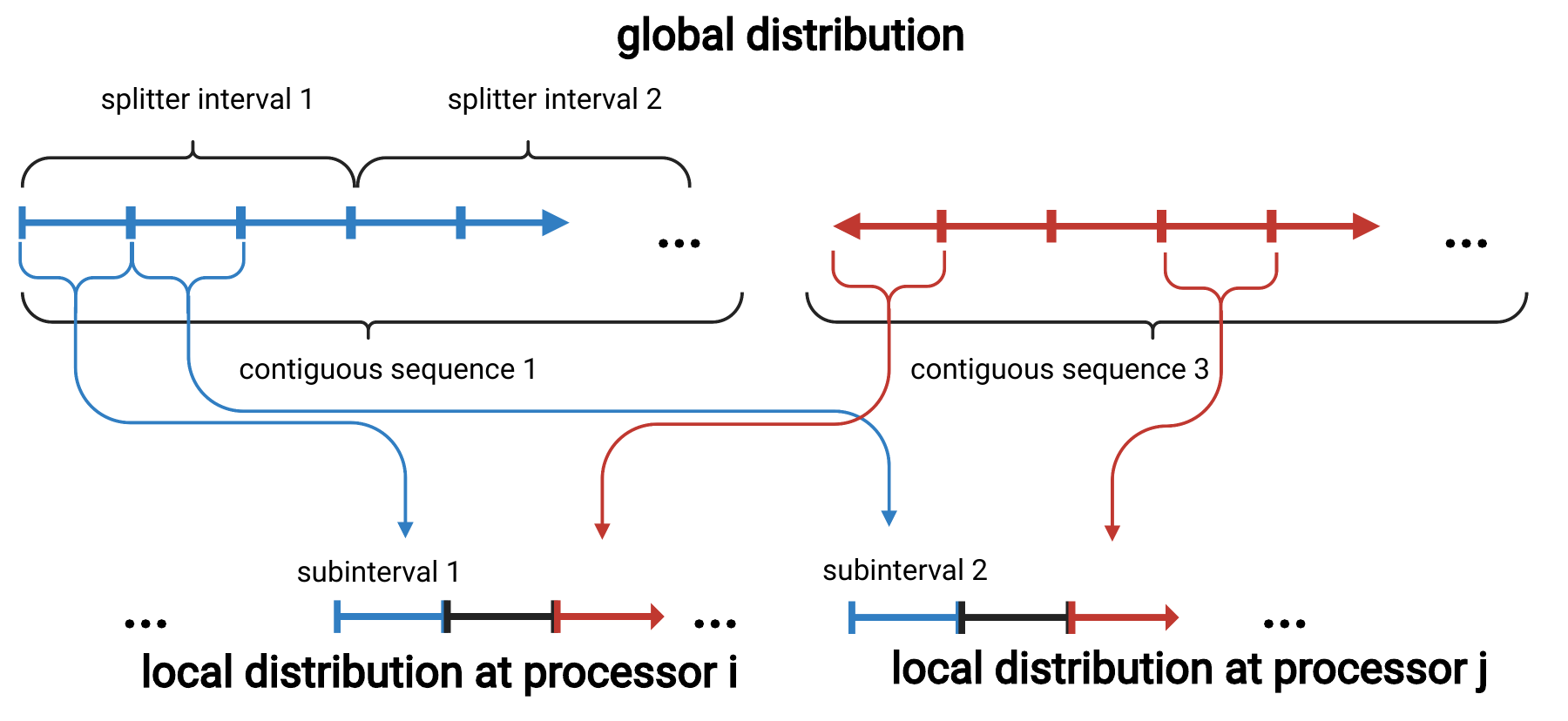}
\caption{The input distribution we use to prove hardness. The distribution is obtained by assuming a global sorted order of input keys (at the top) and then constructing the local distribution at each process from the global order (below).} \label{fig:inputDist}
\end{center}
\end{figure}

We choose the input distribution $\Tilde{\mathcal{D}}_T$ by partitioning the global order of keys into $C$ parts,
each containing $\frac{p}{C}$ consecutive intervals, and assigning each processor a single random subinterval from each part.
The distribution (Figure~\ref{fig:inputDist}) is thus a 3-level partition of input into $C$ parts, containing a total of $C \cdot \frac{p}{C}$ intervals, and $C \cdot \frac{p}{C} \cdot C$ subintervals.
We describe the process to generate this distribution in detail below (illustrated in Figure~\ref{fig:inputDist}).
\begin{construction}[input distribution]
\label{inputDistribution}
Given a sorted input of $N$ keys, we follow the steps below to assign the global input to each processor, which gives us the input distribution:
\begin{enumerate}
    \item divide the global input evenly into $p$ intervals;
    \item divide the $p$ intervals into $C$ parts with each part containing $\frac{p}{C}$ intervals (a part is just a sequence of intervals);
    \item divide each interval evenly into $C$ subintervals so that each part has $p$ subintervals; 
    \item for every part, assign exactly one subinterval randomly to each processor.
\end{enumerate}
\end{construction}
The merit of the chosen distribution is that each processor owns keys from a single random interval in that part.
Consequently, if the part is unsampled (the global positions of all keys within it are unknown), a given processor may be able to tell which of its keys belong to the part, but cannot tell which interval in the part they belong to.
More generally, we refer to sequences of unsampled intervals within a single part as a \textbf{contiguous sequence}.
\begin{definition}[contiguous sequence]
\label{contiguousSequence}
 A contiguous sequence is defined to be a contiguous sequence of unsampled intervals of input in the global order such that any processor contains at most one subinterval from the continguous sequence in its local input. 
\end{definition}
We note that in the first round, every part is a contiguous sequence of $\frac{p}{C}$ intervals ($p$ subintervals).
The next lemma shows that our distribution ensures that sampling in different contiguous intervals is independent.
\begin{lemma}
\label{LemmaDistProp}
At any round, for any deterministic histogram-based algorithm, subintervals in each contiguous sequence have the same probability to be sampled by the algorithm. Consequently, intervals in each contiguous sequence have the same probability to be sampled.
\end{lemma}
\begin{proof}
Given the input distribution above, for any contiguous sequence, because each processor contains at most one subinterval and no matter how we interchange these subintervals among processors they would still appear in the same position of the their respective local orderings, for any deterministic algorithm, processors cannot distinguish one subinterval from other subintervals in any assigned contiguous sequence. Moreover, since subintervals in each contiguous sequence are distributed randomly, they are sampled with the same probability. Therefore, intervals in each contiguous sequence are sampled with the same probability.
\end{proof}

The lemma above essentially claims that throughout $T$ rounds, for any contiguous sequence, processors cannot infer any information globally about that contiguous sequence from its local input. Thus, for any deterministic algorithm, at each round, processors can only make the  decision of whether to sample a random subinterval in a given contiguous sequence. Therefore, any deterministic algorithm is equivalent to some sampling strategy of distributing samples for different contiguous sequences.

\subsection{Bound for One Round/BSP Superstep}



We now derive a sample-size lower bound for 2-balanced partitioning for single round partitioning algorithms, which yields a lower bound on the sample size needed by the sample sort algorithm.
We first characterize the probability that a given interval in a contiguous sequence is unsampled, then use this to bound the sample size needed to ensure all intervals in one or more contiguous sequences are sampled.
We restrict our analysis of algorithms to what contiguous sequences they choose to sample from, as we have established (1) in Lemma~\ref{lem:onetoall_to_sample2}, that no comparison-based algorithm that communicates $O(p)$ total keys per round can complete in fewer steps than the minimum rounds needed by a histogram-based scheme, and (2) in Lemma~\ref{LemmaDistProp} that with our input distribution, processors cannot distinguish intervals within the same contiguous sequence.

\begin{lemma}
\label{LemmaOneInterval}
Given $KQ$ samples (across all processors) from a contiguous sequence with $Q$ intervals, the probability that any given interval in the contiguous sequence is unsampled is at least $e^{-4K}$, if $\frac{K}{C}\leq 1/4$ and $Q\geq 2$.
\end{lemma}
\begin{proof}
By Lemma~\ref{LemmaDistProp}, the sampling strategy of any deterministic algorithm for one contiguous sequence is equivalent to random sampling of subintervals in the sequence.
Let $X_i$ be the event that at least one subinterval in interval $i$ is sampled.
$X_i$ does not occur if and only if each subinterval from interval $i$ falls in the set of $CQ-KQ$ unsampled subintervals from the $CQ$ total subintervals in the contiguous sequence
(we assume, without loss of generality, that all samples are chosen from distinct subintervals).
We then compute the probability $1-\Pr[X_i]$ that the $C$ subintervals in interval $i$ are unsampled,
 which can be formulated as \begin{equation}
\begin{aligned}
          1-\Pr[X_i] 
           &=\prod_{i=0}^{C-1} \frac{CQ-KQ-i}{CQ-i} 
           \geq \bigg(\frac{CQ-KQ-C}{CQ-C}\bigg)^C \\
           &= \big(1-\frac{K}{C}\frac{Q}{Q-1}\big)^C 
           \geq \big(1-\frac{2K}{C}\big)^C \geq e^{\frac{-4KC}{C}} = e^{-4K},
\end{aligned}
\end{equation}
where we replace each term in the product with the lower bound and use $e^{-x} \leq 1 - x + \frac{x^2}{2} \leq 1 - \frac{x}{2} $ if $0 \leq x \leq 1$.
\end{proof}


Having obtained a bound on the probability of an interval remaining unsampled, we now bound the sample size needed to sample all unsampled intervals with a given probability.
\begin{lemma}
\label{LemmaContSeq}
Given $D$ contiguous sequences, each containing $Q$ unsampled intervals, any deterministic algorithm that samples all unsampled intervals in one round, with probability at least $\rho$, needs to sample $QKD$ keys in total, 
where $ K \geq \frac{1}{4}\log \Big(\frac{DQ}{\log(\frac{1}{\rho})}\Big)$ assuming $\frac{K}{C} \leq 1/4$ and $Q \geq 2$.
 \end{lemma}
 \begin{proof}
 
We first characterize the case of one contiguous sequence, that is, $D=1$. Let $X_i$ be the event that interval $i$ is sampled, and we want to compute an upper bound for $\Pr[X_1, X_2, ..., X_Q]$.
We first show that  \begin{equation}
    \forall i, \Pr[X_i|X_1,...,X_{i-1}] \leq \Pr[X_i],
\end{equation}
which can be easily verified since the probability that for some interval $i^{th}$, $j^{th}$ subinterval in the interval is in one of the $CQ-KQ-(j-1)$ unsampled positions out of $CQ-(j-1)$ possible positions is less than the probability that the $j^{th}$ subinterval is in one of the $CQ-KQ-(j-1)$ unsampled positions out of $CQ-(j-1)-(i-1)$ possible positions since at least $i-1$ samples must be in $X_1...X_{i-1}$.

Then, from Lemma \ref{LemmaOneInterval}, we have
\begin{equation}
 \begin{aligned}
       \Pr[X_1, X_2...X_Q] &= \Pr[X_1]P[X_2|X_1]...\Pr[X_Q|X_1,...X_{Q-1}] \\
      &\leq \Pr[X_1]\Pr[X_2]...\Pr[X_Q] \leq \big(1- e^{-4K}\big)^Q \label{eq:onecontigb},
 \end{aligned}
 \end{equation}
 for one contiguous sequence.
 
For $D$ contiguous sequences, the distribution of samples among contiguous sequences must satisfy $Q\sum_{i=1}^D \kappa_i = QKD$, where we sample $\kappa_iQ$ samples from the $i^{th}$ contiguous sequence (note that the total number of samples is $QKD$). 
The optimal sampling strategy is then given by the solution to the optimization problem given by maximizing the probability of all intervals being sampled,
\begin{equation}
 \begin{aligned}
          \max_{K_1,...,K_{D}} & \prod_{i=1}^D \Pr_i[X_1,...,X_Q|\kappa_i]
          \quad \textrm{s.t.} \quad  \sum_{i=1}^D \kappa_i  = K D,  \text{ and }
          \forall i, \kappa_i \geq 0.
 \end{aligned}
 \end{equation}
We upper bound the maximum of the above problem, by using our analysis of a single contiguous sequence (Eq.~\ref{eq:onecontigb}), which gives an upper bound of the above objective function of $\prod_{i=1}^D \big(1- e^{-4\kappa_i}\big)^{Q}$.
By quantifying when this upper bound can reach $\rho$, we obtain a lower bound on sample size.
By standard application of the method of Lagrange multipliers, the above upper bound is maximized when all $\kappa_i=K$.
%
So, if the algorithm samples all unsampled intervals with probability $\rho$, then it must be that
\begin{equation}
 \begin{aligned}
                 \rho &\leq\Big(1- e^{-4K}\Big)^{DQ} \leq e^{-DQe^{-4K}},
 \end{aligned}
 \end{equation}
 which implies
\begin{equation}
    K \geq \frac{1}{4}\log \Big(\frac{DQ}{\log(\frac{1}{\rho})}\Big).
\end{equation}
 \end{proof}
With the above lemma, we now bound the sample-size neeeded to achieve a good partitioning in a single round (sample complexity of sample sort).
 \begin{theorem}
 \label{TheoremOne}
 Any randomized algorithm that achieves $2$-balanced partition, with probability greater than $\rho$, in a single round, requires $\Omega(p\log p)$ samples for large enough $p$. 
 \end{theorem}
 \begin{proof}
By Lemma~\ref{lem:onetoall_to_sample}, if there exists a single round $2$-balanced partitioning algorithm with $r=O(p\log p)$ sample complexity for all $p$, then a single round algorithm with the same complexity also exists for $2$-globally-balanced-partitioning.
Further, by Lemma~\ref{lem:onetoall_to_sample2}, to obtain a round lower-bound, it suffices to consider algorithms that sample/histogram a set of keys in each round.
Finally, by Yao's principle (Corollary~\ref{YaoPrincipleCoro}), it suffices to consider the round complexity of the best deterministic algorithm on the input distribution defined in Construction~\ref{inputDistribution}.
By Lemma \ref{LemmaContSeq}, for any $\rho$, for large enough $p, C, Q$ (it suffices to let $C= Q= \sqrt{p}$), the number of samples needed to find a suitable $2$-globally-balanced partition is
\begin{equation}
 \begin{aligned}
            K &\geq\frac{1}{4}\log \Big(\frac{p}{\log(\frac{1}{\rho})}\Big).\\
 \end{aligned}
 \end{equation} 
So, the algorithm needs at least $K=\Omega(\log(p))$, or $r=\Omega(p\log(p))$ samples in total to achieve $2$-globally-balanced partition in one round.
By the above reductions, this result implies that any single round randomized algorithm for 2-balanced partitioning communicates $\Omega(p \log p)$ total keys.
 \end{proof}

\subsection{Bound for $\Omega(\log^*(p))$ Rounds/BSP Supersteps}

We now consider algorithms with multiple sample-histogram rounds, demonstrating that any BSP algorithm that communicates $O(p)$ total keys per round, requires $\Omega(\log^* p)$ rounds.
To derive our lower bound, we make use of a result from the distribution theory of runs (A.M. Mood, 1940~\cite{mood1940distribution}).
This result allows us to bound how many sequences of subintervals of a certain size remain unsampled in a contiguous sequence after a round of sampling.
\begin{definition}
Let $n^{(k)}=n\cdot(n-1)\cdot(n-2)\ldots (n-k+1)$, $n^{(0)} = 1$.
\end{definition}
\begin{lemma}[distribution theory of runs]
\label{LemmaDistributionTheory}
Consider random arrangements of $m$ element of two kinds, for example $m_1$ $a$'s and $m_2$ $b$'s such that $m_1 + m_2 = m$. Let $r_{1i}$ denote the total number of runs  of $a$'s of length $i$. Let $s_{1k} = \sum_{i=k}^{m_1} r_{1i}$ denote the total number of runs of $a$'s of length at least $k$. We have $\mathbb{E}[s_{1k}] = \frac{(m_2+1)m_1^{(k)}}{m^{(k)}}$ and $\Var[s_{1k}] = \frac{(m_2+1)^{(2)}m_1^{(2k)}}{m^{(2k)}} + \frac{(m_2+1)m_1^{(k)}}{m^{(k)}}\big(1-\frac{(m_2+1)m_1^{(k)}}{m^{(k)}}\big)$.
\end{lemma}
\begin{proof}
Refer to equation \textit{(3.13), (3.15)} in \textit{The Distribution Theory of Runs} \cite{mood1940distribution} by \textit{A. M. Mood}.
\end{proof}
Lemma~\ref{LemmaDistributionTheory} immediately allows us to bound the number of subsequences of unsampled intervals after a sample-histogram round on a given contiguous sequence.
\begin{corollary}
\label{CoroDistributionTheory}
After $K_iQ_i$ random subintervals are sampled from a contiguous sequence containing $CQ_i$ subintervals, the number of maximal subsequences of unsampled subintervals of size at least $n$, $\hat{X}_{i+1}$, 
satisfies $\mathbb{E}[\hat{X}_{i+1}] = \frac{(K_iQ_i+1){(CQ_i-K_iQ_i)^{(n)}}}{(CQ_i)^{(n)}}$ and $\Var[\hat{X}_{i+1}] = \frac{(K_i Q_i+1)^{(2)}((C-K_i)Q_i)^{(2n)}}{(CQ_i)^{(2n)}}+\mathbb{E}[\hat{X}_{i+1}](1-\mathbb{E}[\hat{X}_{i+1}])$.
\end{corollary}
\begin{proof}
This result follows directly from Lemma~\ref{LemmaDistributionTheory}, since we can treat each subinterval as an element and whether a subinterval is sampled as whether a element $b$ appears in the arrangement, and the problem of characterizing $\hat{X}_{i+1}$ can be thought of as the problem of characterizing the number of runs of $a$'s of size at least $n$ in the random arrangements of $CQ_i$ elements with ($CQ_i-K_iQ_i$) $a$'s and $K_i Q_i$ $b$'s.
\end{proof}

We now apply this bound in a setting with many contiguous subsequences.
We aim to bound the number of contiguous subsequences at the $i^{th}$ step (random variable denoted by $X_{i}$).
 \begin{definition}
Let 
     $\alpha_i(K_i) = \frac{K_i Q_i + 1}{e^{4K_i(Q_{i+1}+1)}}.$
 \end{definition}
%
\begin{lemma}
\label{LemmaOneStep}
Consider $D_i$ contiguous sequences, such that each contiguous sequence contains $Q_i$ intervals, and a total of at most $K_{i}D_iQ_i$ subintervals are sampled from these contiguous sequences (randomly within each contiguous sequence), if $\frac{Q_{i+1}}{Q_i} \leq \frac{1}{2}$ and  $\frac{K_i}{C} \leq \frac{1}{4}$, the number of maximal subsequences of unsampled intervals of size at least $Q_{i+1}$ (for any distribution of samples across the $D_i$ contiguous sequences), $X_{i+1}$, satisfies \begin{equation}
    \mathbb{E}[X_{i+1}] \geq  D_i \alpha_i(K_i)
.\end{equation} 
Further, the probability that $X_{i+1}$ is smaller than half of the expected value is bounded as follows,
\begin{equation}
    \Pr[X_{i+1} \leq  \frac{1}{2} D_i \alpha_i(K_i)] \leq \frac{4}{D_i \alpha_i(K_i)}
.\end{equation} 
\end{lemma}
\begin{proof} 
Let $\kappa_jQ_i$ be the number of samples from contiguous sequence $j$, with $\sum_{j=1}^{D_i}\kappa_j=D_iK_i$.
Then, by Lemma~\ref{CoroDistributionTheory}, the number of maximal subsequences in contiguous sequence $j$ of unsampled subintervals of size at least $n$ satisfies $\mathbb{E}[\hat{X}_{i+1}(\kappa_j) ] = \frac{(\kappa_jQ_i+1){(CQ_i-\kappa_jQ_i)^{(n)}}}{(CQ_i)^{(n)}}$ and $\Var[\hat{X}_{i+1}(\kappa_j)] = \frac{(\kappa_j Q_i+1)^{(2)}((C-\kappa_j)Q_i)^{(2n)}}{(CQ_i)^{(2n)}}+\mathbb{E}[\hat{X}_{i+1}](1-\mathbb{E}[\hat{X}_{i+1}])$.
We pick $n=C(Q_{i+1}+1)$, which ensures that the sequence of subintervals contains the entirety of $Q_{i+1}$ intervals, and hence no subintervals in these intervals are sampled.
Further, $\mathbb{E}[\hat{X}_{i+1}] \geq \alpha_i(\kappa_j)$ since, by similar analysis as in the one-round case,
\begin{equation}
\begin{aligned}
    \frac{\mathbb{E}[\hat{X}_{i+1}(\kappa_j)]}{\kappa_jQ_i + 1} &= \prod_{j=0}^{n-1} \frac{CQ_i-\kappa_jQ_i - j}{CQ_i -j} \\
   \geq \bigg(1-\frac{\kappa_jQ_i}{CQ_i-n}\bigg)^{n}
    &\geq \bigg(1-\frac{\kappa_jQ_i}{CQ_i-CQ_{i+1}}\bigg)^{n}\\
    \geq \bigg(1-\frac{\kappa_jQ_i}{CQ_i-CQ_{i}/2}\bigg)^{n} 
    &= \Big(1 - \frac{2\kappa_j}{C}\Big)^n
    \geq e^{\frac{-4\kappa_jn}{C}} = e^{-4\kappa_j(Q_{i+1}+1)}. \notag
\end{aligned}
\end{equation}
Since sampling is independent across all $D_i$ contiguous sequences (since they are maximal, they must be disjoint), the total number of maximal subsequences of unsampled subintervals of size at least $n$, $X_{i+1}$, satisfies
\[
\mathbb{E}[X_{i+1}] = \sum_{j=1}^{D_i}\mathbb{E}[\hat{X}_{i+1}(\kappa_j)] \geq \sum_{j=1}^{D_i} \alpha_i(\kappa_j).
\]
Since, $\sum_{j=1}^{D_i}\kappa_j=D_iK_i$ and $\alpha_i'(x)<0$ for $x>1$, as
\begin{align*}
\alpha_i'(x)&=e^{-4x(Q_{i+1}+1)}(-4Q_i(Q_{i+1}+1)x+Q_i-4(Q_{i+1}-1) \\
&\leq
- e^{-4x(Q_{i+1}+1)}(4Q_{i}(Q_{i+1}+1)),
\end{align*}
$\mathbb{E}[X_{i+1}]$ is minimized when $\kappa_j = K_i$, so $\mathbb{E}[X_{i+1}]\geq D_i\alpha_i(K_i)$.

To quantify the variance, we observe that for all $m,k \leq n$, 
\begin{equation}
    n^{(2k)} \leq \big(n^{(k)}\big)^2 \text{ and }
    \frac{(n-m)^{(2k)}}{n^{(2k)}} \leq \bigg(\frac{(n-m)^{(k)}}{n^{(k)}}\bigg)^2.
\end{equation}
Using the above inequalities and Corollary \ref{CoroDistributionTheory}, the variance for one contiguous sequence is
\begin{equation}
\begin{aligned}
    \Var[\hat{X}_{i+1}(\kappa_i)] &= 
    \mathbb{E}[\hat{X}_{i+1}(\kappa_i)]  \\
    &+ \bigg(\frac{(K_i Q_i+1)^{(2)}((C-K_i)Q_i)^{(2n)}}{(CQ_i)^{(2n)}} - \mathbb{E}^2[\hat{X}_{i+1}(\kappa_i)]\bigg) \\
    &\leq \mathbb{E}[\hat{X}_{i+1}(\kappa_i)].
\end{aligned}
\end{equation}
Further, across $D_i$ sequences, using independence, we have
\begin{equation}
 \begin{aligned}
          \Var[X_{i+1}] &=
          \Var\bigg[\sum_{j=1}^{D_i} \hat{X}_{(i+1)}(\kappa_j) \bigg] = \sum_{j=1}^{D_i} \Var[\hat{X}_{(i+1)}(\kappa_j)]\\
          &\leq \sum_{j=1}^{D_i} \mathbb{E}[\hat{X}_{(i+1)}(\kappa_j)]] \leq \mathbb{E}[X_{(i+1)}]. 
 \end{aligned}
 \end{equation}

\noindent We now obtain the desired bound on deviation from expectation by using Cantelli's inequality,
\begin{equation}
    \Pr[X_{i+1} \leq \mathbb{E}[X_{i+1}]-\lambda] \leq \frac{\Var[X_{i+1}] }{\Var[X_{i+1}] + \lambda^2}.
\end{equation}
Substituting $\lambda = \frac{\mathbb{E}[X_{i+1}]}{2}$, we get
\begin{equation}
    \begin{aligned}
    \Pr&[X_{i+1} \leq  \frac{1}{2} D_i \alpha_i(K_i)] \leq
        \Pr[X_{i+1} \leq  \frac{1}{2}\mathbb{E}[X_{i+1}]]\\ & \leq \frac{\Var[X_{i+1}]}{ \Var[X_{i+1}] + \mathbb{E}^2[X_{i+1}]/4}
           \leq \frac{\mathbb{E}[X_{i+1}]}{ \mathbb{E}[X_{i+1}] + \mathbb{E}^2[X_{i+1}]/4} \\
        &= \frac{4}{4 + \mathbb{E}[X_{i+1}]} \leq \frac{4}{\mathbb{E}[X_{i+1}]}
        \leq \frac{4}{\mathbb{E}[X_{i+1}| \forall j, \kappa_j = K_i ]} = \frac{4}{D_i \alpha_i(K_i)}.
    \end{aligned}\notag
\end{equation}
\end{proof}


\begin{theorem}
\label{TheoremConstantRounds}
Any randomized algorithm, given $O(p)$ samples for each round, 
requires $T = \Omega(\log^*(p))$ number of rounds to achieve $2$-balanced partitioning with probability greater than $1/2$, for large enough $p$.
\end{theorem}
\begin{proof}
We show that any histogram-based sorting algorithm requires more than $T=\frac{\log^*(p)}{6}$ rounds to achieve a constant probability of success with $p$ samples per round\footnote{Note that increasing the sample per round by a constant factor to $(p\alpha)$ per round can lead to at most a $\alpha$ factor decrease in the number of rounds for the lower bound.}.
By the same argument as in the proof of Theorem~\ref{TheoremOne}, this bound also extends to the round complexity of any randomized algorithm that communicates a total of $O(p)$ keys per round, completing the proof.
Our argument proceeds by induction on the round number $i$.
We show that after the $i^{th}$ sample-histogram round that, with probability at least $1-\sum_{j=1}^{i}\frac{4}{D_i\alpha_i(K_i)}$, there are at least $D_{i+1}$ contiguous sequences of \begin{equation}
    Q_{i+1}=\lfloor e^{T-i+1}\rfloor
\end{equation} (unsampled) intervals left, where $D_1=C=\frac{p}{4e^{T}}$ and $D_{i+1}=\frac{D_i\alpha_i(K_i)}{2}$. We choose $K_i$ such that $K_iD_iQ_i = p$ for all $i$. 
We conservatively assume that exactly $D_i$ contiguous sequences are left before the $i^{th}$ round with the above probability as having more than $D_i$ contiguous sequences left can only worsen the complexity of the algorithm.

%
%

For $i=0$ (before the first round), the inductive assumption holds since there are $C$ contiguous (unsampled) sequences each containing $\frac{p}{C}=4e^T>Q_1$ intervals.
For the inductive step, given that prior to the $(i+1)^{th}$ sample-histogram, 
(assuming $K_i\leq C/4$, which we prove next), we apply Lemma~\ref{LemmaOneStep} to lower bound the number of unsampled contiguous sequences for the next round with high probability, namely, we obtain
\begin{align*}
\Pr[X_{i+1} \geq  \frac{1}{2} D_i \alpha_i(K_i)] &\geq \bigg(1-\frac{4}{D_i \alpha_i(K_i)}\bigg)\bigg( 1-\sum_{j=1}^{i-1}\frac{4}{D_j\alpha_j(K_j)}\bigg) \\
&\geq 1-\sum_{j=1}^{i}\frac{4}{D_j\alpha_j(K_j)}.
\end{align*}
To complete the inductive argument, it then suffices to prove $K_i\leq \frac{C}{4}$, given our inductive assumption.

We unravel the inductive relationship to bound $K_i$ and to obtain the final bound on the probability after $T$ steps.
Assume, without loss of generality, that there are exactly $p$ samples per round, then the number of samples per unsampled interval at round $1$ is $K_1=1$.
Further, at round $i+1$,
\begin{equation}
\begin{aligned}
K_{i+1} \ &=\ \frac{p}{D_{i+1}Q_{i+1}} 
   \ =\ \frac{2pe^{4K_i(Q_{i+1}+1)}}{D_i(K_iQ_i+1)}Q_{i+1} \\
        &\leq\ \frac{2pe^{4K_i(Q_{i+1}+1)}}{D_iK_iQ_i}Q_{i+1}
       \ =\  2e^{4K_i(Q_{i+1}+1)}Q_{i+1} \\
        &\leq M(T) e^{K_iM(T)},\quad\text{where}\quad M(T)=4(e^T+1).
\end{aligned}
\end{equation}

\noindent  Let 
\[R_i =  M(T)e^{M(T)e^{M(T)e^{{...}^{M(T)}}}} \Big\}\text{$i$ powers.}\]
For large enough $T$, we have
\begin{equation}
    \begin{aligned}
    \log^* (R_i) &= 1 + \log^* (\log R_i) =  1 + \log^* (\log M(T) + X_{i-1}) \\
    &\leq 1 + \log^* (X_{i-1}^2) < 1 + \log^* (e^{X_{i-1}}) \\
    &\leq 2 + \log^* (X_{i-1}) \leq 2i + \log^* (M(T)) \\
    &\leq 2T + 1 + \log^* (2T) \leq 3T.
    \end{aligned}
\end{equation}
Note that\footnote{$\uparrow \uparrow$ denotes tetration and $\log^*$ denotes super-logarithm which is the inverse operation of tetration.} $\log^* (x) = 1 + \log^* (\log x)$ by definition.
Therefore, since $\log^* (z) = y$ if and only if $z = e \uparrow \uparrow y$, \begin{equation}
\begin{aligned}
K_i &< R_i = e \uparrow \uparrow \log^* (R_i) < e \uparrow \uparrow 3T.\\
\end{aligned}
\end{equation}
Since, $T=\frac{\log^*(p)}{6}$, we have that
\begin{align*}
K_i<e \uparrow \uparrow 3T &< e\uparrow \uparrow (\frac{1}{2} \log^*(p))  \leq e\uparrow \uparrow (\log^* (p) - 2) \\
&= e\uparrow \uparrow (\log^* (\log^2(p))) = \log \log(p).
\end{align*}
Since, $C=\frac{p}{4e^T}$ and $\frac{p}{e^T}>\log \log p$ for sufficiently large $p$, we have shown that $K_i<\frac{C}{4}$.
Further, since $K_iD_iQ_i=p$ for all $i$, we have that
\begin{equation}
\begin{aligned}
D_T &= \frac{K_1D_1Q_1}{K_TQ_T} > \frac{C}{e \uparrow \uparrow 3T} > \frac{C}{\log \log p}.
\end{aligned}
\end{equation}
Since $D_T > \frac{C}{\log \log p} \geq 1$ for large enough $p$, it suffices to apply our inductive bound on probability,
\begin{align*}
\Pr[X_{T} \geq  1] &\geq \Pr[X_{T} \geq  \frac{C}{2\log \log p}] 
\geq \Pr[X_{T} \geq  \frac{1}{2}D_{T-1}\alpha_{T-1}(K_{T-1})] \\
&\geq 1-\sum_{j=1}^{T-1}\frac{4}{D_j\alpha_j(K_j)} 
\geq 1-T\frac{2}{D_T} 
\geq 1-\frac{T \log \log p }{2C} 
\geq \frac{1}{2},
\end{align*}
for sufficiently large $p, C \geq T\log \log p$, which completes the proof.
%
\end{proof}
Note that we have proved a slightly stronger theorem than what we need to show optimality of the new algorithm.
We have shown a lower bound of $\Omega(\log^* p)$ rounds when the algorithm samples $O(p\log^*(p))$ keys in total, instead of $O(p)$.

\section{Discussion and related work}
Parallel sorting is an extensively studied algorithm, both in the experimental and theoretical contexts.

\noindent \textbf{Multiple levels of partitioning and data movement}: To alleviate the bottleneck due to the data partitioning step or in the final data-exchange (where $O(p^2)$ messages have to be exchanged), some algorithms~\cite{sundar2013hyksort,axtmann2015practical,hssSPAA} resort to a 2-level sorting algorithm where the algorithm first groups together several processors into a "processor-group" and first executes parallel sorting at the level of processor-groups and then within each processor group.
In this context, the Histogram partitioning algorithm may be applied to efficiently partition keys across and within processor-groups.

\noindent \textbf{Merge-based sorting networks}: Classical results, such as by Batcher~\cite{batcher1968sorting}, Cole~\cite{cole1988parallel} and AKS~\cite{ajtai19830} use sorting networks and optimize the depth of the network. Our paper deals with large scale parallel sorting which are \textit{partitioning-based}, and move data only once, while typical sorting-networks naively imply many rounds of interprocessor communication. Leighton~\cite{tightBoundsSorting} presented tight lower bounds for sorting networks, but no such bound is known for partitioning-based algorithms.

\noindent \textbf{Experimental validation}: 
The HSS algorithm~\cite{hssSPAA} has been shown to outperform other state-of-the-art parallel sorting implementations, as well as to achieve good scalability and partitioning balance.
We note that Histogram partitioning is faster than HSS~\cite{hssSPAA} by a factor of $\Omega(\log^* p)$ and while the sampling ratio suggested by our work is more efficient, in a practical setting it is unlikely to result in a large improvement over HSS. 

\noindent \textbf{Optimal sample complexity}: 
Note that in Theorem \ref{TheoremConstantRounds}, we have also given a formula to compute the lower bound for $K_i$ (the number of samples at each round) and the total number of samples to achieve $2$-balanced partitioning for any randomized algorithm for every constant round $T = 2, 3, 4...$ similar to Theorem \ref{TheoremOne}. Thus these results may be derived easily in future works. 

\noindent \textbf{Related partitioning problems}: 
The partitioning problem we consider also provides a parallel method for computing an approximate ranking of a sequence~\cite{heckel2018approximate}.
In contrast to~\cite{heckel2018approximate}, our lower bounds assume all-to-one comparisons with each sample, as opposed to minimizing pairwise comparisons.

Communication-efficient parallel sorting algorithms have been the subject of many prior studies.
For example~\cite{asymmetricSorting}, study sorting algorithms with asymmetric read and write costs, while ~\cite{cacheObliviousAlgorithms,lowDepthCacheOblivious} study cache oblivious sorting algorithms.
We believe our input distribution and the technique to applying distribution theory of runs can be a valuable tool for analyzing new parallel sorting algorithms, as well as techniques for other problems such as approximate ranking, \textit{semisort} \cite{gu2015top}, and \textit{quantile computation} \cite{chen2020survey}.  

\section{Conclusion}
Data partitioning is a crucial building block of parallel sorting with several past approaches. Via new theoretical analysis, we introduce a histogramming strategy that achieves a balanced partitioning with minimal sample volume and round complexity ($O(p/\log^* p)$ samples per round, $O(\log^* p)$ rounds and $O(p)$ total samples).
To show optimality, we present a lower bound ($\Omega(\log^* p)$) on the number of rounds required for any randomized algorithm to achieve a $2$-partition with $O(p)$ samples per round.
This analysis provides an asymptotically tight analysis of randomized histogramming, a widely used and practical parallel sorting technique.

\bibliographystyle{plainnat}
\bibliography{paper}

\appendix


\section{Appendix: Additional Pseudocode and Diagrams}

This appendix provides some supplementary details to our main contributions. 
Algorithm~\ref{hss:alg} describes the details of the approach presented in Section~\ref{histopartAlg}.
Figure~\ref{fig:lb_outline} provides a pictorial overview of our lower bound argument in Section~\ref{lowerBound}.

\begin{figure}
{
\centering
\includegraphics[width=0.96\linewidth]{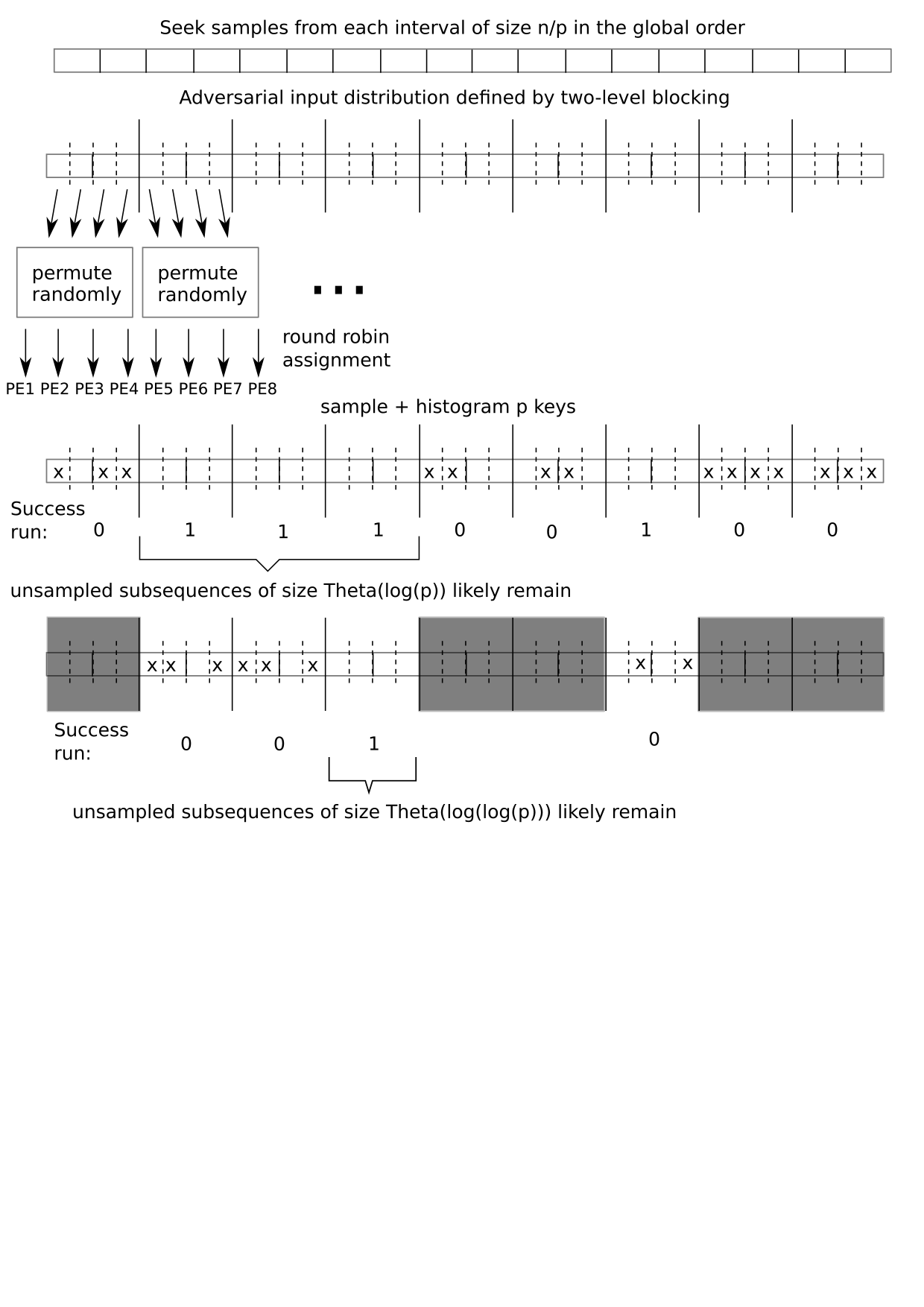}
\vskip -4.5cm
\caption{Overview of the construction of the input distribution used in our proof as well as the connection between rounds of sampling/histogramming and success runs.}
\label{fig:lb_outline}
}
\end{figure}

{\small
\begin{algorithm}
 \SetAlgoNoLine%
  \LinesNumbered
 Processors $\mathcal{P}_0, \mathcal{P}_1 ... \mathcal{P}_p$ start with $\mathcal{P}_i$ owning $n/p$ distinct keys contained $\mathcal{A}_i$ with $\mathcal{A}=\bigcup_{i=1}^p \mathcal{A}_i$\\
 Initialize bounds to include all elements, e.g., $L_k=I(1), U_k=I(N), \ \forall k\in \{1,\ldots, p-1\}$ \\
 Initialize the set of invalid splitters as $W=\{1,\ldots, p-1\}$ \\
 \While{$|W|> 0$}{
 \textbf{At all processors $\mathcal{P}_i$:} \\
   Set the sampling ratio as $\psi=\frac{2p^2}{|W|N\log^*p}$ \\
   Sample each key in $\mathcal{A}_i \cap \Big(\bigcup_{k \in W}[L_k, U_k]\Big)$ with sampling probability $\psi$ to create subset $S_i$\\
   Allgather and merge/sort sample keys to obtain full sample $S=\bigcup_{i=1}^{p-1} S_i$\\
   Compute local histogram $\mathcal{J}$ for $\mathcal{S}$\\
   Reduce all local histograms (accumulate $\mathcal{J}$ to obtain $\mathcal{H}$ on $\mathcal{P}_0$)\\
 \textbf{At processor $\mathcal{P}_0$:} \\
  Receive combined reduced histogram $\mathcal{H}$ for $\mathcal{S}$ from all $\mathcal{P}_i$\\
  Update lower/upper bound keys $L_k$ and $U_k$ for all partitions $k$\\ 
  Set $s_k$ to $L_k$ or $U_k$, depending on whether $R(L_k)$ or $R(U_k)$ is closer to $Nk/p$ for all $k$ \\
  Update $W$ by removing indices of all splitters that are valid within $\epsilon$ error (the $k^{th}$ splitter is valid if $R(s_k)\in [k\frac{N}{p}, k(1+\epsilon)\frac{N}{p}]$) \\
  Broadcast updates to $W$ and $\{L_k,U_k\}_{k=1}^{p-1}$ to all $\mathcal{P}_i$ for next iteration\\
 \textbf{At all processors $\mathcal{P}_i$:} \\
   Receive updates to $\{L_k,U_k\}_{k=1}^{p-1}$ and $W$ from $\mathcal{P}_0$ \\
 }
  Collect (if necessary) and output splitters $s_1,\ldots s_{p-1}$
  \caption{HSS partitioning algorithm, steps are numbered in the order they are executed}
  \label{hss:alg}
 \end{algorithm}
}

\clearpage

%




\end{document}